\documentclass[conference,a4paper]{IEEEtran}

\usepackage[english]{babel}
\usepackage[utf8]{inputenc}
\usepackage[T1]{fontenc}
\usepackage{hyphenat}

\usepackage{times}
\usepackage{amsmath,amssymb,amstext,amsfonts,mathrsfs,amsthm,mathtools,cases}
\usepackage{comment}
\DeclarePairedDelimiter\ceil{\lceil}{\rceil}

\usepackage{graphicx}                           
\graphicspath{{.}{./Figures/}}
\usepackage[tight]{subfigure}                   

\usepackage[nolist,printonlyused]{acronym}      

\usepackage[sort]{cite}                         

\usepackage{enumerate}                          

\usepackage{multirow}
\usepackage{tabularx}
\usepackage{booktabs}
\newcolumntype{C}{>{\centering\arraybackslash}X}
\newcolumntype{R}{>{\flushright\arraybackslash}X}
\newcolumntype{L}{>{\flushleft\arraybackslash}X}

\usepackage{algorithmic}
\usepackage[ruled]{algorithm}
\algsetup{indent=2ex}
\usepackage{eqparbox}

\usepackage{setspace}
\usepackage{xspace} 
\newcommand*{\unit}[2]{\mbox{\ensuremath{#1\,\mathrm{#2}}}}

\usepackage{float}
\floatname{algorithm}{\footnotesize Algorithm}

\usepackage{url}

\usepackage[normalem]{ulem}

\newcommand*{\figref}[1]{Figure~\ref{#1}}
\newcommand*{\secref}[1]{Section~\ref{#1}}
\newcommand*{\tabref}[1]{Table~\ref{#1}}
\newcommand*{\algref}[1]{Algorithm~\ref{#1}}

\theoremstyle{definition}
\newtheorem{thm}{Theorem}

\DeclareMathOperator*{\argmax}{arg\,max}

\usepackage{color}

\newcommand{\comm}[1]{}

\IEEEoverridecommandlockouts

\allowdisplaybreaks[4]

\begin{document}
\date{\today}

\bstctlcite{IEEEexample:BSTcontrol}


\title{Distributed Spectral Efficiency Maximization in Full-Duplex Cellular Networks}

\author{
Jos\'{e} Mairton B.~da Silva Jr$^{\star}$\thanks{Jos\'{e}
Mairton B. da Silva Jr. would like to acknowledge CNPq, a Brazilian research-support agency. The simulations were 
performed on resources 
provided by the Swedish National Infrastructure for Computing (SNIC) at PDC Centre for High 
Performance Computing (PDC-HPC).}, Yuzhe Xu$^{\star}$, G\'{a}bor Fodor$^{\star\dagger}$,
Carlo
Fischione${^\star}$\\
$^\star$KTH, Royal Institute of Technology, Stockholm, Sweden \\
$^\dagger$Ericsson Research, Kista, Sweden
\\[0pt]\vspace{-0.95\baselineskip}
}


\begin{acronym}[LTE-Advanced]
  \acro{2G}{Second Generation}
  \acro{3-DAP}{3-Dimensional Assignment Problem}
  \acro{3G}{3$^\text{rd}$~Generation}
  \acro{3GPP}{3$^\text{rd}$~Generation Partnership Project}
  \acro{4G}{4$^\text{th}$~Generation}
  \acro{5G}{5$^\text{th}$~Generation}
  \acro{AA}{Antenna Array}
  \acro{AC}{Admission Control}
  \acro{AD}{Attack-Decay}
  \acro{ADSL}{Asymmetric Digital Subscriber Line}
  \acro{AHW}{Alternate Hop-and-Wait}
  \acro{AMC}{Adaptive Modulation and Coding}
  \acro{AP}{Access Point}
  \acro{APA}{Adaptive Power Allocation}
  \acro{ARMA}{Autoregressive Moving Average}
  \acro{ATES}{Adaptive Throughput-based Efficiency-Satisfaction Trade-Off}
  \acro{AWGN}{Additive White Gaussian Noise}
  \acro{BB}{Branch and Bound}
  \acro{BD}{Block Diagonalization}
  \acro{BER}{Bit Error Rate}
  \acro{BF}{Best Fit}
  \acro{BFD}{bidirectional full duplex}
  \acro{BLER}{BLock Error Rate}
  \acro{BPC}{Binary Power Control}
  \acro{BPSK}{Binary Phase-Shift Keying}
  \acro{BRA}{Balanced Random Allocation}
  \acro{BS}{base station}
  \acro{CAP}{Combinatorial Allocation Problem}
  \acro{CAPEX}{Capital Expenditure}
  \acro{CBF}{Coordinated Beamforming}
  \acro{CBR}{Constant Bit Rate}
  \acro{CBS}{Class Based Scheduling}
  \acro{CC}{Congestion Control}
  \acro{CDF}{Cumulative Distribution Function}
  \acro{CDMA}{Code-Division Multiple Access}
  \acro{CL}{Closed Loop}
  \acro{CLPC}{Closed Loop Power Control}
  \acro{CNR}{Channel-to-Noise Ratio}
  \acro{CPA}{Cellular Protection Algorithm}
  \acro{CPICH}{Common Pilot Channel}
  \acro{CoMP}{Coordinated Multi-Point}
  \acro{CQI}{Channel Quality Indicator}
  \acro{CRM}{Constrained Rate Maximization}
	\acro{CRN}{Cognitive Radio Network}
  \acro{CS}{Coordinated Scheduling}
  \acro{CSI}{Channel State Information}
  \acro{CUE}{Cellular User Equipment}
  \acro{D2D}{Device-to-Device}
  \acro{DCA}{Dynamic Channel Allocation}
  \acro{DE}{Differential Evolution}
  \acro{DFT}{Discrete Fourier Transform}
  \acro{DIST}{Distance}
  \acro{DL}{downlink}
  \acro{DMA}{Double Moving Average}
	\acro{DMRS}{Demodulation Reference Signal}
  \acro{D2DM}{D2D Mode}
  \acro{DMS}{D2D Mode Selection}
  \acro{DPC}{Dirty Paper Coding}
  \acro{DRA}{Dynamic Resource Assignment}
  \acro{DSA}{Dynamic Spectrum Access}
  \acro{DSM}{Delay-based Satisfaction Maximization}
  \acro{ECC}{Electronic Communications Committee}
  \acro{EFLC}{Error Feedback Based Load Control}
  \acro{EI}{Efficiency Indicator}
  \acro{eNB}{Evolved Node B}
  \acro{EPA}{Equal Power Allocation}
  \acro{EPC}{Evolved Packet Core}
  \acro{EPS}{Evolved Packet System}
  \acro{E-UTRAN}{Evolved Universal Terrestrial Radio Access Network}
  \acro{ES}{Exhaustive Search}
  \acro{FD}{full duplex}
  \acro{FDD}{frequency division duplex}
  \acro{FDM}{Frequency Division Multiplexing}
  \acro{FER}{Frame Erasure Rate}
  \acro{FF}{Fast Fading}
  \acro{FSB}{Fixed Switched Beamforming}
  \acro{FST}{Fixed SNR Target}
  \acro{FTP}{File Transfer Protocol}
  \acro{GA}{Genetic Algorithm}
  \acro{GBR}{Guaranteed Bit Rate}
  \acro{GLR}{Gain to Leakage Ratio}
  \acro{GOS}{Generated Orthogonal Sequence}
  \acro{GPL}{GNU General Public License}
  \acro{GRP}{Grouping}
  \acro{HARQ}{Hybrid Automatic Repeat Request}
  \acro{HD}{half-duplex}
  \acro{HMS}{Harmonic Mode Selection}
  \acro{HOL}{Head Of Line}
  \acro{HSDPA}{High-Speed Downlink Packet Access}
  \acro{HSPA}{High Speed Packet Access}
  \acro{HTTP}{HyperText Transfer Protocol}
  \acro{ICMP}{Internet Control Message Protocol}
  \acro{ICI}{Intercell Interference}
  \acro{ID}{Identification}
  \acro{IETF}{Internet Engineering Task Force}
  \acro{ILP}{Integer Linear Program}
  \acro{JRAPAP}{Joint RB Assignment and Power Allocation Problem}
  \acro{UID}{Unique Identification}
  \acro{IID}{Independent and Identically Distributed}
  \acro{IIR}{Infinite Impulse Response}
  \acro{ILP}{Integer Linear Problem}
  \acro{IMT}{International Mobile Telecommunications}
  \acro{INV}{Inverted Norm-based Grouping}
	\acro{IoT}{Internet of Things}
  \acro{IP}{Integer Programming}
  \acro{IPv6}{Internet Protocol Version 6}
  \acro{ISD}{Inter-Site Distance}
  \acro{ISI}{Inter Symbol Interference}
  \acro{ITU}{International Telecommunication Union}
  \acro{JAFM}{joint assignment and fairness maximization}
  \acro{JAFMA}{joint assignment and fairness maximization algorithm}
  \acro{JOAS}{Joint Opportunistic Assignment and Scheduling}
  \acro{JOS}{Joint Opportunistic Scheduling}
  \acro{JP}{Joint Processing}
	\acro{JS}{Jump-Stay}
  \acro{KKT}{Karush-Kuhn-Tucker}
  \acro{L3}{Layer-3}
  \acro{LAC}{Link Admission Control}
  \acro{LA}{Link Adaptation}
  \acro{LC}{Load Control}
  \acro{LOS}{line of sight}
  \acro{LP}{Linear Programming}
  \acro{LTE}{Long Term Evolution}
	\acro{LTE-A}{\ac{LTE}-Advanced}
  \acro{LTE-Advanced}{Long Term Evolution Advanced}
  \acro{M2M}{Machine-to-Machine}
  \acro{MAC}{medium access control}
  \acro{MANET}{Mobile Ad hoc Network}
  \acro{MC}{Modular Clock}
  \acro{MCS}{Modulation and Coding Scheme}
  \acro{MDB}{Measured Delay Based}
  \acro{MDI}{Minimum D2D Interference}
  \acro{MF}{Matched Filter}
  \acro{MG}{Maximum Gain}
  \acro{MH}{Multi-Hop}
  \acro{MIMO}{Multiple Input Multiple Output}
  \acro{MINLP}{mixed integer nonlinear programming}
  \acro{MIP}{Mixed Integer Programming}
  \acro{MISO}{Multiple Input Single Output}
  \acro{MLWDF}{Modified Largest Weighted Delay First}
  \acro{MME}{Mobility Management Entity}
  \acro{MMSE}{Minimum Mean Square Error}
  \acro{MOS}{Mean Opinion Score}
  \acro{MPF}{Multicarrier Proportional Fair}
  \acro{MRA}{Maximum Rate Allocation}
  \acro{MR}{Maximum Rate}
  \acro{MRC}{Maximum Ratio Combining}
  \acro{MRT}{Maximum Ratio Transmission}
  \acro{MRUS}{Maximum Rate with User Satisfaction}
  \acro{MS}{Mode Selection}
  \acro{MSE}{Mean Squared Error}
  \acro{MSI}{Multi-Stream Interference}
  \acro{MTC}{Machine-Type Communication}
  \acro{MTSI}{Multimedia Telephony Services over IMS}
  \acro{MTSM}{Modified Throughput-based Satisfaction Maximization}
  \acro{MU-MIMO}{Multi-User Multiple Input Multiple Output}
  \acro{MU}{Multi-User}
  \acro{NAS}{Non-Access Stratum}
  \acro{NB}{Node B}
	\acro{NCL}{Neighbor Cell List}
  \acro{NLP}{Nonlinear Programming}
  \acro{NLOS}{non-line of sight}
  \acro{NMSE}{Normalized Mean Square Error}
  \acro{NORM}{Normalized Projection-based Grouping}
  \acro{NP}{non-polynomial time}
  \acro{NRT}{Non-Real Time}
  \acro{NSPS}{National Security and Public Safety Services}
  \acro{O2I}{Outdoor to Indoor}
  \acro{OFDMA}{Orthogonal Frequency Division Multiple Access}
  \acro{OFDM}{Orthogonal Frequency Division Multiplexing}
  \acro{OFPC}{Open Loop with Fractional Path Loss Compensation}
	\acro{O2I}{Outdoor-to-Indoor}
  \acro{OL}{Open Loop}
  \acro{OLPC}{Open-Loop Power Control}
  \acro{OL-PC}{Open-Loop Power Control}
  \acro{OPEX}{Operational Expenditure}
  \acro{ORB}{Orthogonal Random Beamforming}
  \acro{JO-PF}{Joint Opportunistic Proportional Fair}
  \acro{OSI}{Open Systems Interconnection}
  \acro{PAIR}{D2D Pair Gain-based Grouping}
  \acro{PAPR}{Peak-to-Average Power Ratio}
  \acro{P2P}{Peer-to-Peer}
  \acro{PC}{Power Control}
  \acro{PCI}{Physical Cell ID}
  \acro{PDCCH}{physical downlink control channel}
  \acro{PDF}{Probability Density Function}
  \acro{PER}{Packet Error Rate}
  \acro{PF}{Proportional Fair}
  \acro{P-GW}{Packet Data Network Gateway}
  \acro{PL}{Pathloss}
  \acro{PRB}{Physical Resource Block}
  \acro{PROJ}{Projection-based Grouping}
  \acro{ProSe}{Proximity Services}
  \acro{PS}{Packet Scheduling}
  \acro{PSO}{Particle Swarm Optimization}
  \acro{PUCCH}{physical uplink control channel}
  \acro{PZF}{Projected Zero-Forcing}
  \acro{QAM}{Quadrature Amplitude Modulation}
  \acro{QoS}{quality of service}
  \acro{QPSK}{Quadri-Phase Shift Keying}
  \acro{RAISES}{Reallocation-based Assignment for Improved Spectral Efficiency and Satisfaction}
  \acro{RAN}{Radio Access Network}
  \acro{RA}{Resource Allocation}
  \acro{RAT}{Radio Access Technology}
  \acro{RATE}{Rate-based}
  \acro{RB}{resource block}
  \acro{RBG}{Resource Block Group}
  \acro{REF}{Reference Grouping}
  \acro{RF}{Radio-Frequency}
  \acro{RLC}{Radio Link Control}
  \acro{RM}{Rate Maximization}
  \acro{RNC}{Radio Network Controller}
  \acro{RND}{Random Grouping}
  \acro{RRA}{Radio Resource Allocation}
  \acro{RRM}{Radio Resource Management}
  \acro{RSCP}{Received Signal Code Power}
  \acro{RSRP}{Reference Signal Receive Power}
  \acro{RSRQ}{Reference Signal Receive Quality}
  \acro{RR}{Round Robin}
  \acro{RRC}{Radio Resource Control}
  \acro{RSSI}{Received Signal Strength Indicator}
  \acro{RT}{Real Time}
  \acro{RU}{Resource Unit}
  \acro{RUNE}{RUdimentary Network Emulator}
  \acro{RV}{Random Variable}
  \acro{SAC}{Session Admission Control}
  \acro{SCM}{Spatial Channel Model}
  \acro{SC-FDMA}{Single Carrier - Frequency Division Multiple Access}
  \acro{SD}{Soft Dropping}
  \acro{S-D}{Source-Destination}
  \acro{SDPC}{Soft Dropping Power Control}
  \acro{SDMA}{Space-Division Multiple Access}
  \acro{SER}{Symbol Error Rate}
  \acro{SES}{Simple Exponential Smoothing}
  \acro{S-GW}{Serving Gateway}
  \acro{SINR}{signal-to-interference-plus-noise ratio}
  \acro{SI}{self-interference}
  \acro{SIP}{Session Initiation Protocol}
  \acro{SISO}{Single Input Single Output}
  \acro{SIMO}{Single Input Multiple Output}
  \acro{SIR}{Signal to Interference Ratio}
  \acro{SLNR}{Signal-to-Leakage-plus-Noise Ratio}
  \acro{SMA}{Simple Moving Average}
  \acro{SNR}{Signal to Noise Ratio}
  \acro{SORA}{Satisfaction Oriented Resource Allocation}
  \acro{SORA-NRT}{Satisfaction-Oriented Resource Allocation for Non-Real Time Services}
  \acro{SORA-RT}{Satisfaction-Oriented Resource Allocation for Real Time Services}
  \acro{SPF}{Single-Carrier Proportional Fair}
  \acro{SRA}{Sequential Removal Algorithm}
  \acro{SRS}{Sounding Reference Signal}
  \acro{SU-MIMO}{Single-User Multiple Input Multiple Output}
  \acro{SU}{Single-User}
  \acro{SVD}{Singular Value Decomposition}
  \acro{TCP}{Transmission Control Protocol}
  \acro{TDD}{time division duplex}
  \acro{TDMA}{Time Division Multiple Access}
  \acro{TNFD}{three node full duplex}
  \acro{TETRA}{Terrestrial Trunked Radio}
  \acro{TP}{Transmit Power}
  \acro{TPC}{Transmit Power Control}
  \acro{TTI}{Transmission Time Interval}
  \acro{TTR}{Time-To-Rendezvous}
  \acro{TSM}{Throughput-based Satisfaction Maximization}
  \acro{TU}{Typical Urban}
  \acro{UE}{user equipment}
  \acro{UEPS}{Urgency and Efficiency-based Packet Scheduling}
  \acro{UL}{uplink}
  \acro{UMTS}{Universal Mobile Telecommunications System}
  \acro{URI}{Uniform Resource Identifier}
  \acro{URM}{Unconstrained Rate Maximization}
  \acro{VR}{Virtual Resource}
  \acro{VoIP}{Voice over IP}
  \acro{WAN}{Wireless Access Network}
  \acro{WCDMA}{Wideband Code Division Multiple Access}
  \acro{WF}{Water-filling}
  \acro{WiMAX}{Worldwide Interoperability for Microwave Access}
  \acro{WINNER}{Wireless World Initiative New Radio}
  \acro{WLAN}{Wireless Local Area Network}
  \acro{WMPF}{Weighted Multicarrier Proportional Fair}
  \acro{WPF}{Weighted Proportional Fair}
  \acro{WSN}{Wireless Sensor Network}
  \acro{WWW}{World Wide Web}
  \acro{XIXO}{(Single or Multiple) Input (Single or Multiple) Output}
  \acro{ZF}{Zero-Forcing}
  \acro{ZMCSCG}{Zero Mean Circularly Symmetric Complex Gaussian}
\end{acronym}

\maketitle
\IEEEpeerreviewmaketitle

\newcommand*{\BS}[1]{\ensuremath{\text{BS}_{#1}}}
\newcommand*{\UE}[1]{\ensuremath{\text{UE}_{#1}}}

\begin{abstract}
Three-node full-duplex is a promising new transmission mode between a full-duplex
capable wireless node and two other wireless nodes that use half-duplex transmission
and reception respectively.
Although three-node full-duplex transmissions can increase the
spectral efficiency without requiring full-duplex capability of user devices,
inter-node interference -- in addition to the inherent self-interference --
can severely degrade the performance.
Therefore, as methods that provide effective self-interference mitigation evolve, the
management of inter-node interference is becoming increasingly important.
This paper considers a cellular system in which a full-duplex capable base station
serves a set of half-duplex capable users.
As the spectral efficiencies achieved
by the uplink and downlink transmissions are inherently intertwined, the objective
is to device channel assignment and power control algorithms that maximize the
weighted sum of the uplink-downlink transmissions. 
To this end a distributed auction based channel assignment algorithm is proposed, in which the
scheduled uplink users and the base station jointly determine the set of downlink users for
full-duplex transmission.
Realistic system simulations indicate that the spectral
efficiency can be up to \unit{89}{\%} better than using
the traditional half-duplex mode.
Furthermore, when the
self-interference cancelling level is high, the impact of the user-to-user interference is severe
unless properly managed.
\end{abstract}


\section{Introduction}\label{sec:intro}

Traditional cellular networks operate in \ac{HD} transmission mode, in which
a \ac{UE} or the \ac{BS} either transmits or receives on any given frequency channel.
However, the increasing demand
to support the transmission of unprecedented data quantities has led the research
community to investigate new wireless transmission technologies.
Recently, in-band \ac{FD} has been proposed as a key enabling technology to
drastically increase the spectral efficiency of conventional wireless transmission modes.
Due to recent advances in antenna design, interference cancellation algorithms,
\ac{SI} suppression techniques and prototyping of \ac{FD} transceivers,
\ac{FD} transmission is becoming a realistic technology component of
advanced wireless -- including cellular -- systems, especially in the low transmit power regime
\cite{Heino2015,Laughlin2015}.

In particular, in-band \ac{FD} and \ac{TNFD} transmission modes can
drastically increase the spectral efficiency of conventional wireless transmission modes since
both transmission techniques have the potential to double
the spectral efficiency of traditional wireless systems operating in \ac{HD}~\cite{Bharadia2013,Thilina2015}.
\ac{TNFD} involves three nodes, but only one of them needs to have \ac{FD} capability.
The FD-capable node transmits to its receiver node while receiving from another transmitter node
on the same frequency channel.

\begin{figure}
\centering
\includegraphics[width=0.6\linewidth,trim=2mm 0mm 2mm 4mm,clip]{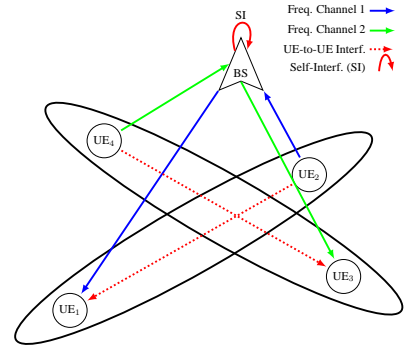}
\caption{A cellular network employing \ac{FD} with two \acp{UE} pairs. The base station selects
pairs of \acp{UE}, represented by the ellipses, and jointly schedules them for \ac{FD}
transmission by allocating frequency channels in the \ac{UL} and \ac{DL}.
To mitigate UE-to-UE interference, it is advantageous to co-schedule DL/UL users for \ac{FD} transmission
that are far apart, such as \UE{1}-\UE{2} and \UE{3}-\UE{4}.}
\label{fig:scenario_pair_fd}
\end{figure}
As illustrated in Figure~\ref{fig:scenario_pair_fd},
\ac{FD} operation in a cellular environment experiences new types of interference, aside from the
inherently present \ac{SI}. 
Because the level of \ac{UE}-to-\ac{UE} interference depends on the \ac{UE} locations
and their transmission powers, coordination mechanisms are needed to mitigate the
negative
effect of the interference on the spectral efficiency of the system~\cite{Goyal2015}.
A key element of such mechanisms is \ac{UE} {\it pairing} and frequency channel selection
that together determine 
which \acp{UE} should be scheduled for simultaneous \ac{UL} and \ac{DL} transmissions on
specific frequency channels.
Hence, it is crucial to design efficient and fair \acl{MAC} protocols and
physical layer procedures capable of supporting adequate pairing mechanisms.
Furthermore, in future cellular networks the idea is to
move from a fully centralized to a more distributed network~\cite{Osseiran2014}, where the
infrastructure of the \ac{BS} can be used to help the \acp{UE} to communicate in a 
distributed manner and reduce the processing burden at the BS, which is further increased by 
\ac{SI} cancellation.

To the best of our knowledge, the only work to consider a distributed approach for \ac{FD}
cellular networks is reported in~\cite{Bai2013}. 
However, the authors tackle the problem of the
\ac{UE}-to-\ac{UE} interference from an information theoretic perspective, without relating to
resource allocation and power control. 
Conversely, some works consider the joint subcarrier and
power allocation problem~\cite{Nam2015} and the joint duplex mode selection, channel allocation,
and power control problem~\cite{Feng2015} in \ac{FD} networks. 
The cellular network model
in~\cite{Nam2015} is applicable to \ac{FD} mobile nodes rather than to networks operating in
\ac{TNFD} mode. 
The work reported in~\cite{Feng2015} considers the case of \ac{TNFD} transmission
mode in a cognitive femto-cell context with bidirectional transmissions from \acp{UE} and develops
sum-rate optimal resource allocation and power control algorithms. 
However, none of these two
works consider a distributed approach for \ac{FD} cellular networks.

In this paper we formulate the joint problem of user pairing (i.e. co-scheduling of
UL and DL simultaneous transmissions on a frequency channel), and UL/DL power control as
a \ac{MINLP} problem, whose objective is to maximize
the overall spectral efficiency of the system. 
Due to the complexity of the \ac{MINLP} problem
proposed, our solution approach relies
on Lagrangian duality and a distributed auction algorithm in which UL users offer bids on
desirable DL users. 
In this iterative auction process, the BS -- as the entity that owns the radio resources -- 
accepts or rejects bids and performs resource assignment. 
This algorithm is tested
in a realistic system simulator that indicates that the bidding process converges
to a near optimal pairing and power allocation.

\section{System Model and Problem Formulation}
\label{sec:sys_mod}

\subsection{System Model}

We consider a single-cell cellular system in which only the \ac{BS} is \ac{FD} capable, while the
\acp{UE} served by the \ac{BS} are only \ac{HD} capable, as illustrated by \figref{fig:scenario_pair_fd}.
In \figref{fig:scenario_pair_fd}, the \ac{BS} is subject to \ac{SI} and the \acp{UE} in the \ac{UL}
(\UE{2} and \UE{4}) cause \ac{UE}-to-\ac{UE} interference to co-scheduled \acp{UE} in the \ac{DL},
that is to \UE{1} and \UE{3} respectively.
The number of \acp{UE} in the \ac{UL} and \ac{DL} is denoted by $I$ and $J$, respectively, which
are constrained by the total number of frequency channels in the system $F$, i.e., $I \leq F$
and $J \leq F$.
The sets of \ac{UL} and \ac{DL} users are denoted by by $\mathcal{I} = \{1,\ldots,I\}$ and
$\mathcal{J} = \{1,\ldots,J\}$ respectively.

We consider frequency flat and slow fading, such that the channels are constant during the
time slot of a scheduling instance and over the frequency channels assigned to co-scheduled users.
Let $G_{ib}$ denote the effective path gain
between transmitter \ac{UE}  $i$ and the \ac{BS}, $G_{bj}$ denote the effective path gain between
the \ac{BS} and the receiving \ac{UE} $j$, and $G_{ij}$ denote the interfering path gain between
the \ac{UL} transmitter UE $i$ and the \ac{DL} receiver UE $j$.
To take into account the
residual \ac{SI} power that leaks to the receiver, we define $\beta$ as the \ac{SI}
cancellation coefficient, such that the \ac{SI} power at the receiver of the \ac{BS} is
$\beta P_j^d$ when the transmit power is $P_j^d$.

The vector of transmit power levels in the \ac{UL} by \ac{UE} $i$ is denoted by $\mathbf{p^u} = [P_1^u
\ldots P_I^u] $, whereas the \ac{DL} transmit powers by the \ac{BS} is denoted by $\mathbf{p^d} =
[P_1^d\ldots P_J^d]$.
As illustrated in \figref{fig:scenario_pair_fd}, the UE-to-UE interference is
much dependent on the geometry of the co-scheduled UL and DL users, i.e., the {\it pairing} of \ac{UL}
and \ac{DL} users.
Therefore, \ac{UE} pairing is a key functions of the system.
Accordingly,
we define the assignment matrix, $\mathbf{X} \in \{0,1\}^{I\times J}$, such that
\begin{equation*}
 x_{ij} =
  \begin{cases}
   1, & \text{if the UL \UE{i} is paired with the DL \UE{j}}, \\
   0, & \text{otherwise.}
  \end{cases}
\end{equation*}

The \ac{SINR} at the \ac{BS} of transmitting user $i$
and the \ac{SINR} at the receiving user $j$ of the \ac{BS} are given by
\begin{equation}\label{eq:sinr_ul_dl}
 \gamma_{i}^u = \frac{ P_i^u G_{ib} }{ \sigma^2 + \sum_{j=1}^J x_{ij} P_j^d\beta}, \;
 \gamma_{j}^d = \frac{ P_j^d G_{bj} }{ \sigma^2 + \sum_{i=1}^I x_{ij} P_i^u G_{ij}},
\end{equation}
respectively,
where $x_{ij}$ in the denominator of $\gamma_{i}^u$
accounts for the \ac{SI} at the BS, whereas $x_{ij}$ in the denominator of $\gamma_{j}^d$
accounts for the UE-to-UE interference caused by \UE{i} to \UE{j}.

Thus, the achievable spectral efficiency for each user is given by the Shannon equation (in bits/s/Hz) for
the \ac{UL} and \ac{DL} as $C_{i}^u = \log_2(1 + \gamma_{i}^u)$ and $C_{j}^d = \log_2(1 +
\gamma_{j}^d)$, respectively.
In addition to the spectral efficiency, we consider weights for the \ac{UL} and \ac{DL} users,
which are denoted by $\alpha_i^u$ and $\alpha_j^d$, respectively. 
The idea behind weighing is that
it allows the system designer to choose between the commonly used sum rate maximization 
and important fairness related criteria such as the well known {\it
path loss compensation} typically employed in the power control of cellular
networks~\cite{Simonsson2008}. 
For the
weights $\alpha_i^u$ and $\alpha_j^d$, we can account for sum rate maximization with $\alpha_i^u =
\alpha_j^d=1$ and for path loss compensation with $\alpha_i^u = G_{ib}^{-1}$
and $\alpha_j^d= G_{bj}^{-1}$.

\subsection{Problem Formulation}\label{sub:prob_form}

Our goal is to jointly consider the assignment of \acp{UE} in the \ac{UL} and \ac{DL}
(\textit{pairing}), while maximizing the weighted sum spectral efficiency of all users.
Specifically, the problem is formulated as
\begin{subequations}\label{eq:weight_prob}
\begin{align}
\underset{\mathbf{X},\mathbf{p}^u,\mathbf{p}^d}{\text{maximize}}\quad
& \sum_{i=1}^I \alpha_i^u C_{i}^u + \sum_{j=1}^J \alpha_j^d C_{j}^d\label{eq:obj_wmax}\\
\text{subject to}\quad & \gamma_{i}^u \geq \gamma_{\text{th}}^u, \; \forall
i,\label{eq:gamma_ul}\\
\quad& \gamma_{j}^d \geq \gamma_{\text{th}}^d, \; \forall j,\label{eq:gamma_dl}\\
\quad& P_{i}^u \leq P_{\text{max}}^u, \; \forall i,\label{eq:power_ul}\\
\quad& P_{j}^d \leq P_{\text{max}}^d, \; \forall j,\label{eq:power_dl}\\
\quad& \sum_{i=1}^I x_{ij} \leq 1, \; \forall j,\label{eq:one_dl_to_ul}\\
\quad& \sum_{j=1}^J x_{ij} \leq 1, \; \forall i,\label{eq:one_ul_to_dl}\\
\quad& x_{ij} \in \{0,1\}, \; \forall i,j.\label{eq:binary_x}
\end{align}
\end{subequations}
The main
optimization variables
are $\mathbf{p}^u$, $\mathbf{p}^d$ and $\mathbf{X}$.
Constraints \eqref{eq:gamma_ul} and \eqref{eq:gamma_dl} ensure a minimum \ac{SINR} to be achieved
in the \ac{DL} and \ac{UL}, respectively.
Constraints \eqref{eq:power_ul} and \eqref{eq:power_dl}
limit the transmit powers whereas constraints \eqref{eq:one_dl_to_ul}-\eqref{eq:one_ul_to_dl}
assure that only one \ac{UE} in the \ac{DL} can share the frequency resource with a \ac{UE} in
the \ac{UL} and vice-versa. Note that constraints~\eqref{eq:gamma_ul}-\eqref{eq:gamma_dl} require 
that the \ac{SINR} targets for both UL and DL be defined a priori.
\comm{Note that if $I\neq J$, some users will not be paired and will
transmit in \ac{HD} mode if there are frequency channels available.}

Problem \eqref{eq:weight_prob} belongs to the category of \ac{MINLP}, which is known for its high
complexity and computational intractability. Thus, to solve problem \eqref{eq:weight_prob} we will
rely on Lagrangian duality, which is described on \secref{sec:lagr_dual}. 
We develop the optimal
power allocation for a 
\ac{UL}-\ac{DL} pair and the optimal closed-form solution for the
assignment,
which can be solved in a centralized manner. However, in future cellular networks the idea is to
move from a fully centralized to a more distributed network~\cite{Osseiran2014}, offloading the
burden on the \ac{BS}. With this objective, in
\secref{sec:dist_auction_sol} we use the optimal power allocation to create a distributed solution
for the assignment between \ac{UL} and \ac{DL} users.

\section{A Solution Approach Based on Lagrangian Duality}
\label{sec:lagr_dual}

From problem \eqref{eq:weight_prob}, we form the {\it partial} Lagrangian function by considering
constraints \eqref{eq:gamma_ul}-\eqref{eq:gamma_dl} and ignoring the integer
\eqref{eq:one_dl_to_ul}-\eqref{eq:binary_x} and power allocation constraints
\eqref{eq:power_ul}-\eqref{eq:power_dl}.
To this end, we introduce Lagrange multipliers
$\boldsymbol{\lambda}^u,\;\boldsymbol{\lambda}^d$, where
the superscript $u$ and $d$ denote the dimensions
of $I$ and $J$, respectively.
The partial Lagrangian is a function of the Lagrange
multipliers and the optimization variables $\mathbf{X},\mathbf{p}^u,\mathbf{p}^d$ as follows:
\begin{align}\label{eq:partial_Lagr}
L(\boldsymbol{\lambda}^{u},\boldsymbol{\lambda}^{d},\mathbf{X},\mathbf{p}^{u},\mathbf{p}^{d})
&\triangleq -\sum_{i=1}^I \alpha_i^u C_{i}^u - \sum_{j=1}^J \alpha_j^d C_{j}^d + \nonumber \\
&\hspace{-2cm}+\sum_{i=1}^I \lambda_i^u \Big(\gamma_{\text{th}}^u - \gamma_{i}^u
\Big) + \sum_{j=1}^J \lambda_j^d \Big(\gamma_{\text{th}}^d - \gamma_{j}^d \Big).
\end{align}

Let $g(\boldsymbol{\lambda}^{u},\boldsymbol{\lambda}^{d})$ denote the dual function obtained by
minimizing the partial Lagrangian \eqref{eq:partial_Lagr} with respect to the variables
$\mathbf{X},\mathbf{p}^u,\mathbf{p}^d$.
That is, the dual function is
\begin{align}
\label{eq:g_dual}
g(\boldsymbol{\lambda}^{u},\boldsymbol{\lambda}^{d}) &= \underset{\scriptscriptstyle
\mathbf{X}\in\mathcal{X}\;\mathbf{p}^u,\mathbf{p}^d\in\mathcal{P}}{\inf} \hspace{-0.5cm}
L(\boldsymbol{\lambda}^{u},\boldsymbol{\lambda}^{d},\mathbf{X},\mathbf{p}^{u},\mathbf{p}^{d}), 
\end{align}
where $\mathcal{X}$ and $\mathcal{P}$ are the set where the assignment and power allocation
constraints are fulfilled, respectively. Notice that we can rewrite the dual as
\begin{align}\label{eq:mod_dual}
\textstyle
g(\boldsymbol{\lambda}^{u},\boldsymbol{\lambda}^{d}) \!= \hspace{-0.5cm}
\underset{\scriptscriptstyle
\mathbf{X}\in\mathcal{X}\;\mathbf{p}^u,\mathbf{p}^d\in\mathcal{P}}{\inf}
\sum\limits_{n=1}^{N}\! \Bigg( q_{i_n}^u (\mathbf{X},\mathbf{p}^u,\mathbf{p}^d) + q_{j_n}^d
(\mathbf{X},\mathbf{p}^u,\mathbf{p}^d)\Bigg),
\end{align}
where we assume $N=I=J$\comm{$N=\max\{I,J\}$} is the maximum number of \ac{UL}-\ac{DL} pairs, 
$i_n$ and $j_n$ are the \ac{UL} and \ac{DL} users of pair $n$, respectively.
\comm{If $I\neq J$, then some users will not be co-scheduled and will transmit alone if there are 
resources available, which means that they will form a virtual pair.}
Moreover,
\begin{subequations}
\begin{align}
q_{i_n}^u (\mathbf{X},\mathbf{p}^u,\mathbf{p}^d) &\triangleq \lambda_{i_n}^u
\Big(\gamma_{\text{th}}^u -\gamma_{i_n}^u \Big) - \alpha_i^u C_{i_n}^u,\\
q_{j_n}^d (\mathbf{X},\mathbf{p}^u,\mathbf{p}^d) &\triangleq \lambda_{j_n}^d
\Big(\gamma_{\text{th}}^d -\gamma_{j_n}^d \Big) - \alpha_{j_n}^d C_{j_n}^d.
\end{align}
\end{subequations}

We can find the infimum of \eqref{eq:mod_dual} if we maximize the \ac{SINR} of the $N$
\ac{UL}-\ac{DL} pairs. Thus, we can write a closed-form expression for the assignment $x_{ij}$ as
follows:
\begin{align}\label{eq:x_sol}
x_{ij}^{\star} =
  \begin{cases}
   1, & \text{if } (i,j)= \underset{i,j}{\argmax}\;
   \Big(q_{i_n}^{u,\text{max}}+q_{j_n}^{d,\text{max}}\Big)
   \\
   0, & \text{otherwise,}
  \end{cases}
\end{align}
where for simplicity we denoted an ordinary pair as $(i,j)$.
Notice that $x_{ij}^{\star}$ and
equation \eqref{eq:x_sol} uniquely associate an \ac{UL} user with a \ac{DL} user.
However,
the solutions are still tied through the \acp{SINR} $\gamma_{i}^u$ and $\gamma_{j}^d$, i.e., the
solution to the assignment problem is still complex and -- through \eqref{eq:x_sol} --
is intertwined with the optimal power allocation.

Since the \acp{SINR} on the \ac{UL} are not separable from those on the \ac{DL},
we cannot analyse them independently.
Consequently, we need to find the powers that jointly minimize \eqref{eq:mod_dual}.
To this end, we first
analyse the dual problem, given by
\begin{subequations}\label{eq:dual_prob}
\begin{align}
& \underset{\boldsymbol{\lambda}^{u},\;\boldsymbol{\lambda}^{u}}{\text{maximize}} & &
g(\boldsymbol{\lambda}^{u},\boldsymbol{\lambda}^{d})\\
& \text{subject to} & & \lambda_i^u,\; \lambda_j^d, \geq 0, \forall i,j,
\end{align}
\end{subequations}
where recall that $g(\boldsymbol{\lambda}^{u},\boldsymbol{\lambda}^{d})$ is the solution of
problem \eqref{eq:g_dual}.
Notice that if constraints \eqref{eq:gamma_ul}-\eqref{eq:gamma_dl} are
fulfilled in the inequality or equality, $\lambda_i^u \Big(\gamma_{\text{th}}^u - \gamma_{i}^u
\Big)$ and $\lambda_j^d \Big(\gamma_{\text{th}}^d - \gamma_{j}^d \Big)$
will be either negative or zero.
If the terms are negative, then $\lambda_i^u$ or $\lambda_j^d$ will be zero.
Thus, the terms with $\lambda_i^u$ and $\lambda_j^d$ will not impact
$g(\boldsymbol{\lambda}^{u},\boldsymbol{\lambda}^{d})$.
Therefore, the dual is easily solved by
assigning zero to $\lambda_i^u$ or $\lambda_j^d$ whose corresponding UL and DL user fulfils the inequalities
\eqref{eq:gamma_ul}-\eqref{eq:gamma_dl}.
If there are users that do not fulfil the inequalities,
the problem is unbounded.

Therefore, we now turn our attention to the power allocation problem, and
-- based on the above considerations on $\lambda_i^u$ or $\lambda_j^d$--,
we formulate the power allocation problem as:
\begin{subequations}\label{eq:min_qi_qj}
\begin{align}
\underset{\mathbf{p}^u,\mathbf{p}^d}{\text{minimize}} \quad & -\sum_{i=1}^I \alpha_i^u C_i^u -
\sum_{j=1}^J \alpha_j^d C_j^d\\
\text{subject to} \quad& \mathbf{p}^u,\mathbf{p}^d \in \mathcal{P}.
\end{align}
\end{subequations}
From Gesbert et al.~\cite{Gesbert2008}, the optimal transmit power allocation will have either
$P_i^u$ or $P_j^d$ equal to $P_{\text{max}}^u$ or $P_{\text{max}}^d$, given that $i$ and $j$ share
a frequency channel and form a pair.
Moreover, from Feng et al.~\cite[Section III.B]{Feng2013}, the
optimal power allocation lies within the admissible area for pair $(i,j)$, where we do not show 
the explicit expressions for the optimal power allocation here due to space limit.

Therefore, with the optimal transmit powers for any given pair $(i,j)$, and with the closed-form
solution for the assignment in Eq.~\eqref{eq:x_sol}, we can solve the dual
problem~\eqref{eq:dual_prob}. To compute the optimal assignment as given by
Eq.~\eqref{eq:x_sol} requires checking $N!$ assignments~\cite[Section 1]{Burkard1999}, or we could
apply the Hungarian algorithm in a fully centralized manner~\cite[Section 3.2]{Burkard1999}, that
has worst-case complexity of $O(N^3)$.

However, we are not interested in such centralized and
demanding solutions that would increase the burden on the \ac{BS}. Since we are in a
network-controlled environment with the \ac{BS}, we use its resources to provide a distributed
solution for the assignment, whereas the power allocation would remain centralized, because
distributed power allocation schemes require too many iterations to converge. 
Therefore, in the next section we reformulate the closed-form
solution in Eq.~\eqref{eq:x_sol} and propose a fully distributed assignment based on Auction
Theory~\cite{Bertsekas1998}.

\section{Distributed Auction Solution}\label{sec:dist_auction_sol}
With the optimal power allocation for a pair $(i,j)$ at hand, and as mentioned in
\secref{sec:lagr_dual}, we are interested in a distributed solution for the closed-form solution
for the assignment in Eq.~\eqref{eq:x_sol} in order to reduce the burden on the \ac{BS} by 
supporting a more distributed system. 
In \secref{sub:prob_reform} we reformulate the
closed-form expression as
an asymmetric assignment problem, whereas \secref{sub:fund_auct} introduces the fundamental
definitions necessary to propose the distributed auction algorithm in
\secref{sub:auct_alg}. 
Furthermore, we give one of the core results in this paper in
\secref{sub:math_prop}, where we show that the number of
iterations of the algorithms is bounded and that the feasible assignment provided at the end is
within a bound of desired accuracy around the optimal assignment. 

\subsection{Problem Reformulation}\label{sub:prob_reform}

We can rewrite the closed form expression~\eqref{eq:x_sol} as an asymmetric assignment problem,
given by
\begin{subequations}
\label{eq:assym_assig_prob}
\begin{align}
\underset{\mathbf{X}}{\text{maximize}} \quad & \sum_{i=1}^I\sum_{j=1}^J c_{ij} x_{ij}\\
\text{subject to} \quad& \sum_{i=1}^I x_{ij} = 1, \; \forall j,\label{eq:dl_to_ul}\\
\quad& \sum_{j=1}^J x_{ij} = 1, \; \forall i,\label{eq:ul_to_dl}\\
\quad& x_{ij} \in \{0,1\}, \; \forall i,j,\label{eq:binary_x_newProb}
\end{align}
\end{subequations}
where $c_{ij} \!= \alpha_i^u C_i^u + \alpha_j^d C_j^d$ for a pair $(i,j)$ assigned to the same
frequency and it can be understood as the benefit of assigning \ac{UL} user $i$ to \ac{DL}
user $j$.
Constraint~\eqref{eq:dl_to_ul} ensures that the \ac{DL} users are \comm{can be}associated 
with one \ac{UL} user.\comm{at most one \ac{UL} user, and since $J>I$, some \ac{DL} users may not 
associate with a \ac{UL} user.}
Similarly, constraint~\eqref{eq:ul_to_dl} ensures that all the \ac{UL} users need to be associated
with a \ac{DL} user.\comm{Notice that if $J\!=I$, inequality on constraint~\eqref{eq:dl_to_ul} 
becomes an equality and there is no other difference in the manner the solution proposed here 
works.} To solve this problem in a distributed manner, we use Auction Theory.

\subsection{Fundamentals of the Auction}\label{sub:fund_auct}
We consider the assignment problem~\eqref{eq:assym_assig_prob}, where we want to
pair $I$ \ac{UL} and $J$ \ac{DL} users on a one-to-one basis, where the benefit for pairing
\ac{UL} user $i$ to \ac{DL} user $j$ is given by $c_{ij}$.
Initially, we assume that $J\!=I$ \comm{$J\!<I$}and $J\!\leq F$\comm{, where the unassigned 
\ac{DL} users will transmit in \ac{HD} mode. Later,}, whereas the set of \ac{DL} users to which 
\ac{UL} user $i$ can be paired is non-empty and is denoted by $\mathcal{A}(i)$.
We define an assignment $\mathcal{S}$ as a set of
\ac{UL}-\ac{DL} pairs $(i,j)$ such that $j\in \mathcal{A}(i)$ for all
$(i,j)\in \mathcal{S}$ and for each \ac{UL} and \ac{DL} user there can be at most one pair
$(i,j)\in \mathcal{S}$,
respectively.
The assignment $\mathcal{S}$ is said to be \textit{feasible} if it contains $I$ pairs;
otherwise the assignment is called \textit{partial}~\cite{Bertsekas1998}.
Lastly, in terms of the
assignment matrix $\mathbf{X}$, the assignment is feasible if
constraints~\eqref{eq:dl_to_ul}-\eqref{eq:ul_to_dl} are fulfilled for all $i\!\in\mathcal{I}$ and
$j\!\in\mathcal{J}$.

An important notion for the correct operation of the auction algorithm is the
$\epsilon$-\textit{complementary slackness} ($\epsilon$-CS), which relates a partial assignment
$\mathcal{S}$ and a price vector $\mathbf{\hat{p}}\!=[\hat{p}_1 \ldots \hat{p}_J]$. In practice,
the
\ac{DL} user $j$ that supports more interference from a \ac{UL} user $i$ will get a higher price,
i.e., the prices reflect how much a \ac{UL} user $i$ is willing pay to connect to \ac{DL} user $j$.
The couple $\mathcal{S}$ and $\mathbf{\hat{p}}$ satisfy
$\epsilon$-CS if for
every pair $(i,j)\in \mathcal{S}$, \ac{DL} user $j$ is within $\epsilon$ of being the best
candidate pair
for \ac{UL} user $i$~\cite{Bertsekas1998}, i.e.,
\begin{equation}
 c_{ij} - \hat{p}_j \geq \max_{k\in \mathcal{A}(i)} \{ c_{ik} - \hat{p}_k \} - \epsilon, \forall
 (i,j) \in  \mathcal{S}.
\end{equation}
For the sake of clarity we define $c_{ij} - \hat{p}_j$ as the utility 
that \ac{UL} user $i$ can obtain from \ac{DL} user $j$.
The auction algorithm is iterative, where each iteration starts with a partial assignment and the
algorithm terminates when a feasible assignment is obtained.

The iteration process consists of two phases: the \textit{bidding} and the \textit{assignment}.
In the bidding phase, each \ac{UL} user bids for a \ac{DL} user that maximizes the associated
utility ($c_{ij} - \hat{p}_j$), and the
\ac{BS} evaluates the bid received from the \ac{UL} users.
In the assignment phase, the \ac{BS}
(responsible for the transmission to \ac{DL} users), selects the \ac{UL} user with the highest bid
and updates the prices.
In fact, this bidding and assignment process implies that
the \ac{UL} users select the \ac{DL} users which they will be paired with. However,
the information exchange occurs between \ac{UL} users and the \ac{BS},
rather than between the \ac{UL}-\ac{DL} users.
Therefore, we propose a forward auction in which the \ac{UL} users and
the \ac{BS} determines the pairing of \ac{UL}-\ac{DL} users in a distributed manner, where the
\ac{UL} is responsible for bidding and the \ac{BS} for the assignment phase.

In the following, we present some necessary definitions for the iteration process.
First, we define $v_i$
as the maximum utility achieved by \ac{UL} user $i$ on the set of possible \ac{DL} users
$\mathcal{A}(i)$,
which is given by
\begin{equation}\label{eq:max_util_ul_ue}
 v_i = \underset{j\in \mathcal{A}(i)}{\max} \{c_{ij} -\hat{p}_j \}.
\end{equation}
The selected \ac{DL} user $j_i$ is the one that maximizes $v_i$, which is given by
\begin{equation}\label{eq:sel_dl_ue}
 j_i = \underset{j\in \mathcal{A}(i)}{\argmax}\; v_i.
\end{equation}
The best utility offered by other \ac{DL} users than the selected $j_i$ is denoted by
$w_i$ and is given by
\begin{equation}\label{eq:sel_2nd_dl_ue}
 w_i = \underset{j\in \mathcal{A}(i),j\neq j_i}{\max} \{c_{ij} -\hat{p}_j \}.
\end{equation}
The bid of \ac{UL} user $i$ on resource $j_i$ is given by
\begin{equation}\label{eq:bid_dl_ue}
 b_{ij_i} = c_{ij} - w_i + \epsilon.
\end{equation}
Let $\mathcal{P}(j)$ denote the set of \ac{UL} users from which \ac{DL} user $j$
received a bid.
In the assignment phase the prices are updated based on the highest bid received, which is
given by
\begin{equation}\label{eq:upd_prices}
 \hat{p}_{j} = \underset{i\in \mathcal{P}(j)}{\max} \{b_{ij}\}.
\end{equation}
Subsequently, the \ac{BS} adds the pair $(i_j,i)$ to the assignment $\mathcal{S}$, where $i_j$ 
refers to the \ac{UL} user $i\in\mathcal{P}(j)$ that maximizes $b_{ij}$ in 
eq.~\eqref{eq:upd_prices} for \ac{DL} user $j$.
At \ac{UL} user $i$, $\mathbf{\hat{P}}_i = [\hat{P}_{i1}\ldots \hat{P}_{iJ}]$ denotes the price
vector of associating with \ac{DL} user $j$ and it is informed by the \ac{BS}.
Differently from $\hat{P}_{ij}$, $\hat{p}_j$ is the up-to-date maximum price of \ac{DL} user $j$.

Summarizing, in the bidding phase the \ac{UL} users need to evaluate $v_i$, $j_i$, $w_i$ and
$b_{ij_i}$, whereas in the assignment phase the \ac{BS} receives the bids and decides to update the
prices $\hat{p}_j$ or not. In the next section, we propose the distributed auction that is
executed in an asynchronous manner, where the \ac{UL} users perform the bidding, and the
\ac{BS} performs the assignment.
\subsection{The Distributed Auction Algorithm}\label{sub:auct_alg}
\algref{alg:dist_auct_ul} and \algref{alg:dist_auct_bs} show the steps of the iterative process of the
bidding and assignment phases, where the bidding is performed at each \ac{UL} user and the
assignment at the \ac{BS}.
We define messages M1, M2, M3 and M4 that enable the exchange of
information between the \ac{UL} users and the \ac{BS}.
Message M1 informs \ac{UL} user $i$
that the bid was accepted.
Message M2 informs that the bid is not high enough and it also
contains the most updated price $\hat{p}_j$ of the demanded \ac{DL} user $j_i$.
Message M3 informs
that a feasible assignment was found, which allows the auction algorithm to terminate. Message M4
informs the \ac{UL} and \ac{DL} users their respective pairs and transmitting powers.
Notice that
all these messages can be exchanged between the \ac{BS} and the \ac{UL}/\ac{DL} users using
control channels, such as \ac{PUCCH} and \ac{PDCCH}~\cite{3gpp.36.300}. 

\begin{algorithm}
 \footnotesize
  \caption{Distributed Auction Bidding at \ac{UL} user $i$}\label{alg:dist_auct_ul}
 \begin{algorithmic}[1]
   \STATE \textbf{Input:} $\mathbf{c}_{i},\epsilon$\label{alg_line:input_ul}
   \STATE Define $\mathbf{P}_i=0$, $i_j = \varnothing$ and
   $\mathcal{A}(i)=\mathcal{J}$\label{alg_line:eval_ul}
   \WHILE{ Message M3 is not received }\label{alg_line:cont_m3}
      \IF{Message M2 is received}
        \STATE Disconnect from previous \ac{DL} user $j_i$ and set
        $j_i=\emptyset$\label{alg_line:disc_prev_dl}
        \STATE Update prices $\hat{P}_{ij} = \hat{p}_j$\label{alg_line:upd_recv_price}
      \ENDIF
      \item[] \textbf{Bidding Phase at the \ac{UL} users}
	  \IF{$j_i=\emptyset$}
	     \STATE Evaluate $v_i$ according to
	     Equation~\eqref{eq:max_util_ul_ue}\label{alg_line:eval_vi}
	     \STATE Select the \ac{DL} user $j_i$ according to Equation~\eqref{eq:sel_dl_ue}
	     \STATE Evaluate $w_i$ according to
	     Equation~\eqref{eq:sel_2nd_dl_ue}\label{alg_line:sec_best_dl}
	     \STATE Evaluate the bid $b_{ij_i}$ according to
	     Equation~\eqref{eq:bid_dl_ue}\label{alg_line:calc_bid}
	     \STATE \textbf{Report} the selected \ac{DL} user $j_i$ and the bid
	     $b_{ij_i}$ and \textbf{wait} response\label{alg_line:report_bid}
	     \IF{Message M1 is received}
	       \STATE Store the assigned \ac{DL} user $j_i$
	     \ENDIF
	   \ENDIF
   \ENDWHILE
   \STATE Message M4 is received with the assigned \ac{DL} user $j_i$ and the power $p_i$
 \end{algorithmic}
\end{algorithm}

\ac{UL} user $i$ requires as inputs the benefits $\mathbf{c}_{i}$ and the $\epsilon$ for
the bidding phase (see line~\ref{alg_line:input_ul} on~\algref{alg:dist_auct_ul}).
Then, the price vector is initialized with zero, as well as the associated \ac{DL} user
is initially empty and the set of \ac{DL} users it can associate with is the set $\mathcal{J}$ (see line~\ref{alg_line:eval_ul}
in ~\algref{alg:dist_auct_ul}).
The auction algorithm at the \ac{UL} users will continue until message
M3 is received (see line~\ref{alg_line:cont_m3} on~\algref{alg:dist_auct_ul}).
If message M2 is
received, \ac{UL} user $i$ disconnects from the previously associated \ac{DL} user $j_i$ and set
it to $\emptyset$.
Next, the price received from the \ac{BS} is updated (see
lines~\ref{alg_line:disc_prev_dl}-\ref{alg_line:upd_recv_price} on~\algref{alg:dist_auct_ul}).
Notice that message M2 implies
that either the previously associated \ac{DL} user has a new association or the bid was lower than
the current price (the bid was placed with an outdated price).

Subsequently, if
\ac{UL} user $i$ is not associated with a \ac{DL} user, the bidding phase
starts.
In this phase,
the necessary variables -- $v_i$, $j_i$, $w_i$ and $b_{ij_i}$ -- are evaluated (see
lines~\ref{alg_line:eval_vi}-\ref{alg_line:calc_bid} on~\algref{alg:dist_auct_ul}).
\ac{UL} user $i$ reports the
selected \ac{DL} user and bid to the \ac{BS} and wait for the response on
line~\ref{alg_line:report_bid} on~\algref{alg:dist_auct_ul}.
If the response is message M1, then
the association to the selected \ac{DL} user $j_i$ is stored.


\begin{algorithm}
 \footnotesize
  \caption{Distributed Auction Assignment at the \ac{BS}}\label{alg:dist_auct_bs}
 \begin{algorithmic}[1]
   \STATE Estimate the channel gains $G_{ib}$, $G_{bj}$ and $G_{ij}$ for all \ac{UL} and \ac{DL}
     users \label{alg_line:est_g_matrix} 
   \STATE Evaluate the optimal power allocation $\mathbf{p}^u, \mathbf{p}^d$ for every pair
     $(i,j)$ based on the solution of problem \eqref{eq:min_qi_qj}\label{alg_line:eval_power}
   \STATE Evaluate $c_{ij}$ and then send to all \ac{UL} users its respective row
   ($\mathbf{c}_i = [c_{i1}\ldots c_{iJ}]^T$) along with $\epsilon$ \label{alg_line:eval_benef}
   \STATE Initialize the selected \ac{UL} users $i_j=\varnothing , \;\forall j$ and
   $\mathcal{P}(j)=\mathcal{I},\; \forall j$\label{alg_line:init_ul_user}
   \STATE Initialize the prices $\hat{p}_j=0,\;\forall j$ and the assignment matrix
   $\mathbf{X}=\mathbf{0}$\label{alg_line:init_prices}
   \item[]\textbf{Assignment Phase at the \ac{BS} }
   \WHILE{ $\mathbf{X}$ is not feasible } \label{alg_line:cont_condition}
     \IF{receive request from \ac{UL} user $i$}\label{alg_line:bid_accept}
     	\IF[Bid accepted]{$b_{ij} - \hat{p}_j \geq \epsilon$}
     	  \STATE Update prices according to Equation~\eqref{eq:upd_prices}
     	  \STATE Report M2 to the previous assigned user $i_j$ and the updated
     	  prices\label{alg_line:rep_upd_prices}
     	  \STATE Update $i_j = i$ and report M1 to \ac{UL} user $i$\label{alg_line:upd_dl_user}
     	  \STATE Update assignment $\mathbf{X}$ and if feasible, report
     	  M3 to \ac{UL} user $i_j$\label{alg_line:upd_assign}
     	\ELSE
     	  \STATE \textbf{Report} M2 and the updated prices $\hat{p}_j$ to \ac{UL} user
     	  $i$\label{alg_line:report_m2}
     	\ENDIF
     \ELSE
        \STATE Remain assigned to the \ac{UL} user on iteration $t$\label{alg_line:remain_assign}
     \ENDIF
   \ENDWHILE
   \STATE \textbf{Output}: $\mathbf{X},\mathbf{p}^u,\mathbf{p}^d$
   \STATE Report M4 to the \ac{UL} and \ac{DL} user their respective pairs and powers
 \end{algorithmic}
\end{algorithm}

The \ac{BS} runs \algref{alg:dist_auct_bs} and initially needs to acquire or estimate
all channel gains from \ac{UL} and \ac{DL} users, which can be done using reference signals
similar to those standardized by \ac{3GPP}~\cite{3gpp.36.300}.
Next, the BS evaluates the optimal power allocation
$\mathbf{p}^u, \mathbf{p}^d$ for all possible pairs based on the solution of
problem~\eqref{eq:min_qi_qj} (see lines~\ref{alg_line:est_g_matrix}-\ref{alg_line:eval_power} on
\algref{alg:dist_auct_bs}).
Then, the assignment benefits $c_{ij}$ are evaluated and the
corresponding row of each \ac{UL} user $i$ is sent (see line~\ref{alg_line:eval_benef} on
\algref{alg:dist_auct_bs}).
The value of $\epsilon$ is fixed and sent on
line~\ref{alg_line:eval_benef} of \algref{alg:dist_auct_bs}.
The selected \ac{UL} users $i_j$ for all \ac{DL} users $j$ is initially empty, and the set of
possible \ac{UL} users that a
\ac{DL} user may associate is defined as the set of \ac{UL} users $\mathcal{I}$ (see
line~\ref{alg_line:init_ul_user} on \algref{alg:dist_auct_bs}).
The prices $\hat{p}_j$ and the
assignment matrix $\mathbf{X}$ are initialized with zero (see line~\ref{alg_line:init_prices}).

The assignment phase at the \ac{BS} continues until the assignment matrix $\mathbf{X}$ is not
feasible (see line~\ref{alg_line:cont_condition}).
If the \ac{BS} receives request from \ac{UL}
user $i$ and the bid is accepted, the prices are updated based on the new bid and the \ac{BS}
reports M2 to the previously assigned user $i_j$ with the updated prices (see
lines~\ref{alg_line:cont_condition}-\ref{alg_line:rep_upd_prices} on~\algref{alg:dist_auct_bs}).
Then, the \ac{BS} updates the assigned user, reports message M1 to
\ac{UL} user $i$ and update the assignment $\mathbf{X}$ (see
lines~\ref{alg_line:upd_dl_user}-\ref{alg_line:upd_assign} on~\algref{alg:dist_auct_bs}).
If the new assignment is feasible, the BS reports M3 to all \ac{UL} users.

However, if the bid proposed by \ac{UL} user $i$ is not accepted, then the BS reports M2 to
\ac{UE} $i$ with the
updated prices (see line~\ref{alg_line:report_m2} on~\algref{alg:dist_auct_bs}).
Notice that while
the \ac{BS} does not received requests, the assignment does not change (see
line~\ref{alg_line:remain_assign}).
Once a feasible assignment is found and message M3 is sent,
the algorithm has as outputs the matrix assignment $\mathbf{X}$ and the power vectors
$\mathbf{p}^u$ and $\mathbf{p}^d$. With the assignment and the power vectors, message M4 is sent
to \ac{UL} and \ac{DL} users with their respective pairs and powers.

Therefore, by using Algorithms~\ref{alg:dist_auct_ul} and \ref{alg:dist_auct_bs}, we solve in a
distributed manner problem~\eqref{eq:assym_assig_prob}, but it is important to know how many
iterations the algorithms execute until a partial assignment is found, and how far this assignment
is from the optimal solution.
In order to address all these questions, in
\secref{sub:math_prop} we show that the algorithms terminate within a bounded number of iterations
and that the assignment given at the end is within $I\epsilon$ of being optimal.

\vspace{-1mm}
\subsection{Complexity and Optimality}\label{sub:math_prop}
In this subsection we derive a bound on the number of iterations of our proposed distributed
auction algorithms in Theorem~\ref{thm:compl_auction}. Moreover, in
Theorem~\ref{thm:opt_auction} we show that the given assignment solution by
Algorithms~\ref{alg:dist_auct_ul} and \ref{alg:dist_auct_bs} is within $I\epsilon$ of being
optimal.

\vspace{-1mm}
\begin{thm}\label{thm:compl_auction}
 Consider $I$ \ac{UL} users and $J$ \ac{DL} users in a \ac{TNFD} network. The distributed auction
 algorithms~\ref{alg:dist_auct_ul} and \ref{alg:dist_auct_bs} terminate within a finite number of
 iterations bounded by $IJ^2\ceil{\Delta/\epsilon}$, where $\Delta = \underset{\forall i,j}{\max}
 c_{ij} - \underset{\forall i,j}{\min} c_{ij}$.
\end{thm}
\begin{proof}
The proof of this theorem is along the lines of Xu et al.~\cite[Chapter 5]{Xu2014}, where we do not
provide the complete proof herein due to the lack of space.
\end{proof}
\vspace{-2mm}

Theorem~\ref{thm:compl_auction} shows that our algorithms terminate in a finite number of
iterations bounded by $IJ^2\ceil{\Delta/\epsilon}$. However, we still need to know how far the
solution is from the optimal assignment. 
In Theorem~\ref{thm:opt_auction} we show that the
feasible assignment at the end of the distributed auction is within $I\epsilon$ of being
optimal, and if the benefits $c_{ij}$ are integer and $\epsilon<1/J$, the solution is optimal.

\vspace{-1mm}
\begin{thm}\label{thm:opt_auction}
 Consider problem~\eqref{eq:assym_assig_prob}. The distributed auction
 algorithm described by Algorithms~\ref{alg:dist_auct_ul} and \ref{alg:dist_auct_bs} terminate
 with a feasible assignment that is within $I\epsilon$ of being optimal. This feasible assignment
 is optimal if $c_{ij},\forall i,j$ is integer and $\epsilon<1/J$.
\end{thm}\vspace{-4mm}
\begin{proof}
Once more, the proof of this theorem is along the lines of~\cite[Chapter 5]{Xu2014}.
\end{proof}
\vspace{-2mm}

Based on Theorems~\ref{thm:compl_auction} and \ref{thm:opt_auction}, the distributed auction
solution proposed by Algorithms~\ref{alg:dist_auct_ul} and \ref{alg:dist_auct_bs} terminate within
$IJ^2\ceil{\Delta/\epsilon}$ iterations, and in addition the feasible assignment at the end of the
algorithms is within $I\epsilon$ of being optimal. 
Notice that in practice the benefits
$c_{ij}$ are seldom integer, implying that our solution to problem~\eqref{eq:assym_assig_prob} is 
near optimal. 
Moreover, since the primal problem~\eqref{eq:weight_prob} is \ac{MINLP}, the duality gap
between the primal and dual solution is not zero, i.e., we should also take into account the
duality gap on top of the gap between the distributed auction and the optimal assignment for
problem~\eqref{eq:assym_assig_prob}. However, as we show in \secref{sec:num_res_disc}, the gap
between the exhaustive primal solution and the distributed auction is small.

\section{Numerical Results and Discussion}\label{sec:num_res_disc}
In this section we consider a single cell system operating in the urban micro
environment~\cite{3gpp.36.814}. 
The maximum number of frequency channels is $F=25$ that corresponds to the number of
available frequency channel blocks in the 
a 5 MHz \ac{LTE} system~\cite{3gpp.36.814}.
The total number of served \acs{UE} varies between $I+J=8...50$, where we assume that $I=J$.
We set the weights $\alpha_i^u$ and $\alpha_j^d$ based
on a \textit{path loss compensation rule}, where $\alpha_i^u = G_{ib}^{-1}$ and $\alpha_j^d=
G_{bj}^{-1}$. The parameters of this system are set according to \tabref{tab:sim_param}.

To evaluate the performance of the distributed auction in this environment, we use the RUdimentary
Network Emulator (RUNE) as a basic platform for system simulations 
and extended it to \ac{FD} cellular networks.
The RUNE \ac{FD} simulation tool allows to generate the environment of \tabref{tab:sim_param}
and perform Monte Carlo simulations using either an exhaustive search algorithm to solve
problem~\eqref{eq:weight_prob} or the distributed auction.

\begin{table}
 \caption{Simulation parameters}\label{tab:sim_param}
 \scriptsize
 \begin{tabularx}{\columnwidth}{l|l}
    \hline 
    \textbf{Parameter}                         & \textbf{Value} \\
    \hline 
    Cell radius                                & \unit{100}{m} \\
    Number of \ac{UL} \acp{UE} $[I=J]$         & $[4\;5\;6\; 25]$  \\
    Monte Carlo iterations                     & $400$ \\ \hline
    Carrier frequency                          & \unit{2.5}{GHz} \\
    System bandwidth                           & \unit{5}{MHz} \\
    Number of freq. channels $[F]$             & $[4\;5\;6\; 25]$  \\
    \ac{LOS} path-loss model                   & $34.96 + 22.7\log_{10}(d)$\\
    \ac{NLOS} path-loss model                  & $33.36 + 38.35\log_{10}(d)$\\
    Shadowing st. dev. \ac{LOS} and \ac{NLOS}  & \unit{3}{dB} and \unit{4}{dB}\\
    Thermal noise power $[\sigma^2]$           & \unit{-116.4}{dBm}/channel \\
    \ac{SI} cancelling level $[\beta]$         & $[-70\;-100 -110]$~\unit{}{dB}\\ \hline
    Max power $[P^u_{\text{max}}]=[P^d_{\text{max}}]$& \unit{24}{dBm}\\
    Minimum \ac{SINR} $[\gamma_{\text{th}}^u=\gamma_{\text{th}}^d]$  & [0]\unit{}{dB} \\
    Step size $[\epsilon]$                     & 0.1\\
    \hline
\end{tabularx}
\end{table}

Initially, we compare the optimality gap between the exhaustive search solution of the primal
problem~\eqref{eq:weight_prob}, named herein as E-OPT, the optimal solution of the dual
problem~\eqref{eq:dual_prob} using the power allocation based on the corner points and the
centralized Hungarian algorithm for the assignment, named herein as C-HUN, and finally the
solution of the dual problem~\eqref{eq:dual_prob} with the optimal power allocation but now with
the distributed auction solution for the assignment, named herein as D-AUC.
In the following, we compare how the distributed auction solution performs in comparison with a
\ac{HD} system, named herein as HD, and also a basic \ac{FD} solution with random assignment and
equal power allocation (EPA) for \ac{UL} and \ac{DL} users, named herein as R-EPA. Notice that
since in \ac{HD} systems two different time slots are required to serve all the \ac{UL} and
\ac{DL} users, which implies that the sum spectral efficiency is divided by two.


In \figref{fig:CDF_comp_Opt_SE_Total} we show the sum spectral efficiency between E-OPT, C-HUN and
the proposed D-AUC as a measure of the optimality gap. We assume a small system with reduced
number of users, 4 \ac{UL} and \ac{DL} users, and frequency channels, where we increase its number
from 4 to 8. 
Moreover, we consider a \ac{SI} cancelling level of $\beta=$\unit{-100}{dB}.
\begin{figure}
\centering
\includegraphics[width=0.7\linewidth,trim=0mm 0mm 0mm 2mm,,clip]{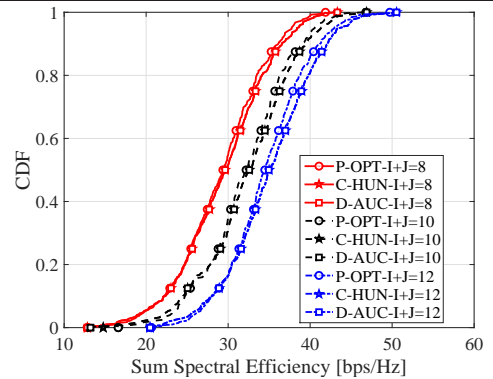}
\caption{\ac{CDF} of the optimality gap of the sum spectral efficiency for different users' load.
We notice that the optimality gap between the proposed D-AUC and the exhaustive search solutions
of the primal and dual, E-OPT and C-HUN, is low, which suggests that we can use a distributed solution
based on the dual and still be close to the centralized optimal solution of the primal.}
\label{fig:CDF_comp_Opt_SE_Total}
\end{figure}
We notice that the differences between the exhaustive search solutions, either E-OPT OR C-HUN, to
the D-AUC is negligible, where in some cases the D-AUC achieves a higher performance than E-OPT
due to lack of computational power to find the best powers. \figref{fig:CDF_comp_Opt_SE_Total} clearly
shows that the optimality gap is low for the distributed auction when compared to the centralized dual
solution (C-HUN) and also to the primal solution (E-OPT). Therefore, we can use a distributed solution
to solve the primal problem~\eqref{eq:weight_prob} and still achieve a solution close to the centralized
optimal solution.


\figref{fig:CDF_comp_sumSE} shows the sum spectral efficiency between the current HD system, a 
naive
\ac{FD} implementation named R-EPA, and the proposed distributed solution D-AUC. We assume a small system
fully loaded with 25 \ac{UL}, \ac{DL} users, and frequency channels, where we analyse the impact of the solutions
for different \ac{SI} cancelling levels of \unit{-110}{dB} and \unit{-70}{dB}, i.e., $\beta
=$\unit{-110}{dB} and \unit{-70}{dB}.
\begin{figure}
\centering
\includegraphics[width=0.7\linewidth,trim=0mm 0mm 0mm 2mm,,clip]{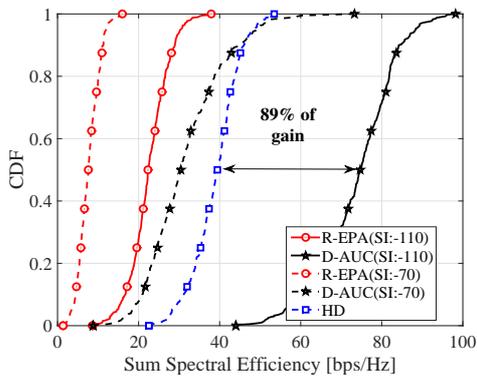}
\caption{\ac{CDF} of the sum spectral efficiency among all users for different \ac{SI} cancelling
levels. We notice that with $\beta =$\unit{-110}{dB} the \ac{UE}-to-\ac{UE} interference is the
limiting factor, where D-AUC mitigates this new interference and outperforms the \ac{HD} mode and
the R-EPA. When $\beta =$\unit{-70}{dB}, the \ac{SI} is the limiting factor, where the mitigation
of the \ac{UE}-to-\ac{UE} interference is not enough to bring gains to \ac{FD} cellular networks.}
\label{fig:CDF_comp_sumSE}
\end{figure}
Notice that with a \ac{SI} cancelling level of \unit{-110}{dB} we achieve \unit{89}{\%}
relative gain in the spectral efficiency
at the 50th percentile, which is close to the expected doubling of \ac{FD} networks. 
Moreover, the naive R-EPA performs
approximately \unit{43}{\%} worse than the HD mode, which shows that despite the high \ac{SI} cancelling level, we do not
have any gain of using \ac{FD} networks. 
This behaviour shows that we should also optimize the \ac{UL}-\ac{DL} pairing and
the power allocation of the \ac{UL} user and of the \ac{BS}. 
When the \ac{SI} level is \unit{-70}{dB}, HD outperforms the
D-AUC and R-EPA with a relative gain of approximately \unit{23}{\%} and \unit{81}{\%} at the 50th percentile, respectively.
This means that with
low \ac{SI} cancelling levels the D-AUC algorithm is not able to overcome the high self-interference,
although the difference to HD is not high. 
As for the R-EPA, notice that its performance is even worse than before, which once more
indicates that we should not use naive implementation of user pairing and power allocation on
\ac{FD} cellular networks. 
Overall, we notice that when the \ac{SI} cancelling level is high, the
\ac{UE}-to-\ac{UE} interference is the limiting factor, where our proposed D-AUC outperforms a
naive \ac{FD} implementation that disregards this interference. 
When the \ac{SI} cancelling level
is low, then the \ac{SI} is the limiting factor, where optimizing the \ac{UE}-to-\ac{UE}
interference is not enough to bring gains to \ac{FD} cellular networks.

\vspace{-1mm}
\section{Conclusion}\label{sec:concl}
In this paper we considered the joint problem of user pairing and
power allocation in \ac{FD} cellular networks.
Specifically, our objective was to maximize the weighted sum spectral efficiency of the users, where
we can tune the weights to sum maximization or path loss compensation.
This problem was posed as a mixed integer nonlinear optimization, which is hard to solve directly,
thus
we resorted to Lagrangian duality and developed a closed-form solution for the assignment and
optimal power allocation. Since we were interested in a distributed solution between the \ac{BS}
and the users, we proposed a novel distributed auction solution to solve the assignment problem,
whereas the power allocation was solved in a centralized manner. We showed that the distributed auction
converges and that it has a guaranteed performance compared to the dual.
The numerical results showed that our distributed solution drastically improved the sum spectral efficiency
in a path loss compensation modelling, i.e. of the users with low spectral efficiency, when compared to current
HD modes when the \ac{SI} cancelling level is high. Furthermore, we noticed that with a high
\ac{SI} cancelling level, the impact of the \ac{UE}-to-\ac{UE} interference is severe and needs to
be properly managed. Conversely, when the \ac{SI} cancelling level is low, a proper management of
the \ac{UE}-to-\ac{UE} interference is not enough to bring gains to \ac{FD} cellular networks.
Studying the case of asymmetric assignment, that is the case of unequal number of UL and DL users 
and the impact of the number of iterations and processing delays in the auction algorithm are left 
for future works.
%
\vspace{-1mm}
\bibliographystyle{IEEEtran}
\bibliography{FDref}

\end{document}